\documentclass[11pt]{article}   %
\usepackage{chao}               %
\usepackage{graphicx}
\usepackage{caption}
\usepackage{subcaption}
\usepackage{paralist}
\usepackage{pdfpages}
\usepackage{xspace}
\usepackage{tcolorbox}
\usepackage{amsmath}

\newcommand{\eps}{\epsilon}

\newcommand{\cP}{\mathcal{P}}

\everypar{\looseness=-1}

\newcommand{\kcut}{\textsc{$k$\--Cut}\xspace}
\newcommand{\mincut}{\textsc{MinCut}\xspace}

\newcommand{\sptree}{\mathcal{T}} %
\newcommand{\treepack}{\tau}      %
\newcommand{\strength}{\sigma}    %
\newcommand{\numcomp}{\kappa}     %
\newcommand{\connectivity}{\lambda} %
\newcommand{\kconn}{\lambda_k}      %

\title{LP Relaxation and Tree Packing for Minimum $k$-cuts\thanks{
       Work on this paper supported in part by NSF grant CCF-1526799.}  }
\author{ Chandra Chekuri\thanks{Department of
    Computer Science, University of Illinois, Urbana, IL
    61801. \texttt{chekuri@illinois.edu}.}
  \and Kent Quanrud\thanks{Department of
    Computer Science, University of Illinois, Urbana, IL
    61801. \texttt{quanrud2@illinois.edu}.}
  \and Chao Xu\thanks{Yahoo! Research, New York, NY
    10003. \texttt{chao.xu@oath.com}. Work done
    while the author was at University of Illinois.}
}
\date{\today}

\begin{document}
\maketitle

\begin{abstract}
  Karger used spanning tree packings \cite{Karger00} to derive a near
  linear-time randomized algorithm for the global minimum cut problem
  as well as a bound on the number of approximate minimum cuts. This
  is a different approach from his well-known random contraction
  algorithm \cite{Karger-thesis95,KargerS96}. Thorup developed a fast
  deterministic algorithm for the minimum $k$-cut
  problem via greedy \emph{recursive } tree packings \cite{Thorup08}.

  In this paper we revisit properties of an LP relaxation for $k$-cut
  proposed by Naor and Rabani \cite{NaorR01}, and analyzed in
  \cite{ChekuriGN06}.  We show that the dual of the LP yields a tree
  packing, that when combined with an upper bound on the integrality
  gap for the LP, easily and transparently extends Karger's analysis
  for mincut to the $k$-cut problem. In addition to the simplicity of
  the algorithm and its analysis, this allows us to improve the
  running time of Thorup's algorithm by a factor of $n$. We also
  improve the bound on the number of $\alpha$-approximate
  $k$-cuts. Second, we give a simple proof that the integrality gap of
  the LP is $2(1-1/n)$. Third, we show that an optimum solution to the
  LP relaxation, for all values of $k$, is fully determined by the
  principal sequence of partitions of the input graph. This allows us
  to relate the LP relaxation to the Lagrangean relaxation approach of
  Barahona \cite{b-00} and Ravi and Sinha \cite{rs-08}; it also shows
  that the idealized recursive tree packing considered by Thorup gives
  an optimum dual solution to the LP. This work arose from an effort
  to understand and simplify the results of Thorup \cite{Thorup08}.
\end{abstract}


\section{Introduction}
\label{sec:intro}
The global minimum cut problem in graphs (\mincut) is well-known and
extensively studied. Given an undirected graph $G=(V,E)$ with
non-negative edge capacities $c: E \rightarrow \mathbb{R}_+$, the goal
is to remove a minimum capacity set of edges such that the residual
graph has at least two connected components. When all capacities are
one, the mincut of a graph is its global edge-connectivity.  The \kcut
problem is a natural generalization. Given a graph $G=(V,E)$ and an
integer $k \ge 2$, the goal is to remove a minimum capacity set of
edges such that the residual graph has at least $k$ connected
components. \mincut and \kcut have been extensively studied in the
literature. Initial algorithms for \mincut were based on a reduction
to the $s$-$t$-mincut problem. However, it was realized later on that
it can be solved more efficiently and directly. Currently the best
deterministic algorithm for \mincut runs in $O(mn + n^2 \log n)$ time
\cite{StoerW} and is based on the maximum adjacency ordering approach
of Nagamochi and Ibaraki \cite{Nagamochi1992}. On the other hand,
there is a near-linear time Monte Carlo randomized algorithm due to
Karger \cite{Karger00}.  Bridging the gap between the running times
for the deterministic and randomized algorithms is a major open
problem. In recent work \cite{kt-15,hrw-17} obtained near-linear time
deterministic algorithms for \emph{simple} unweighted graphs.

The \kcut problem is NP-Hard if $k$ is part of the input
\cite{GoldschmidtH94}, however, there is a polynomial-time algorithm
for any fixed $k$. Such an algorithm was first devised by Goldschmidt
and Hochbaum \cite{GoldschmidtH94}, and subsequently there have been
several different algorithms improving the run-time. The randomized
algorithm of Karger and Stein \cite{KargerS96} runs in
$\tilde{O}(n^{2(k-1)})$ time and outputs the optimum cut with high
probability. The fastest deterministic algorithm, due to Thorup
\cite{Thorup08}, runs in $\tilde{O}(mn^{2k-2})$
time~\cite{Thorup08}. Upcoming work of Gupta, Lee and Li
\cite{gll-18b} obtains a faster run-time of
$\tilde{O}(k^{O(k)}n^{(2\omega/3 + o(1))k})$ if the graph has small
integer weights, where $\omega$ is the exponent in the run-time of
matrix multiplication.
It is also known that \kcut is $W[1]$-hard when
parameterized by $k$ \cite{defpr-03}; that is, we do not expect an
algorithm with a run-time of $f(k)n^{O(1)}$.  Several algorithms that
yield a $2$-approximation are known for \kcut; Saran and Vazirani's
algorithm based on repeated minimum-cut computations gives
$(2-2/k)$-approximation \cite{SaranV95}; the same bound can be
achieved by removing the $(k-1)$ smallest weight edges in a Gomory-Hu
tree of the graph \cite{SaranV95}.
Nagamochi and Kamidoi showed that using the concept of extreme sets,
a $(2-2/k)$-approximation can be found even faster \cite{NagamochiK07}.
Naor and Rabani developed an LP
relaxation for \kcut \cite{NaorR01} and this yields a
$2(1-1/n)$-approximation \cite{ChekuriGN06}. Ravi and Sinha
\cite{rs-08} obtained another $2(1-1/n)$-approximation via a
Lagrangean relaxation approach which was also considered independently
by Barahona \cite{b-00}.  A factor of $2$, for large $k$, is the best
possible approximation under the Small Set Expansion hypothesis
\cite{manurangsi-17}. Recent work has obtained a $1.81$ approximation
in $2^{O(k^2)}n^{O(1)}$ time \cite{gll-18b}; whether a PTAS can be
obtained in $f(k) \text{poly}(n)$ time is an interesting open problem.

\paragraph{Motivation and contributions:} The main motivation for this
work was to simplify and understand Thorup's tree packing based
algorithm for \kcut. Karger's near-linear time algorithm and analysis
for the \mincut problem \cite{Karger00} is based on the well-known
theorem of Tutte and Nash-Williams (on the minmax relation for
edge-disjoint trees in a graph). It is simple and elegant; the main
complexity is in the improved running time which is achieved via a
complex dynamic program. Karger also tightened the bound on the number
of $\alpha$-approximate minimum cuts in a graph (originally shown via
his random contraction algorithm) via tree packings. In contrast to
the case of mincut, the main structural result in Thorup's work on
\kcut is much less easy to understand and motivate.  His proof
consists of two parts.  He shows that an ideal tree packing obtained
via a recursive decomposition of the graph, first outlined in
\cite{Thorup07}, has the property that any optimum $k$-cut crosses
some tree in the packing at most $2k-2$ times. The second part argues
that a greedy tree packing with sufficiently many trees approximates
the ideal tree packing arbitrarily well. The greedy tree packing is
closely related to a multiplicative weight update method for solving a
basic tree packing linear program, however, no explicit LP is used in
Thorup's analysis.  Thus, although Thorup's algorithm is very simple
to describe (and implement), the analysis is somewhat opaque.

In this paper we make several contributions which connect Thorup's
tree packing to the LP relaxation for \kcut \cite{NaorR01}. We outline
the specific contributions below.

\begin{itemize}
\item We show that the dual of LP for \kcut gives a tree packing and
  one can use a simple analysis, very similar to that of Karger, to show that
  any optimum $k$-cut crosses some tree in the packing at most
  $(2k-3)$ times. Thorup proved a bound of $(2k-2)$ for his tree
  packing.
  This leads to a slightly faster algorithm than that
  of Thorup and also to an improved bound on the number
  of approximate $k$-cuts.
\item We give a new and simple proof that the integrality gap of
  the LP for $k$-cut is upper bounded by $2(1-1/n)$. We note that
  the proof claimed in \cite{NaorR01} was incorrect and the proof in
  \cite{ChekuriGN06} is indirect and technical.

\item We show that the optimum solution of the $k$-cut LP, for all
  values of $k$, can be completely characterized by the principal
  sequence of partitions of the cut function of the given graph. This
  establishes the connection between the dual of the LP relaxation and
  the ideal recursive tree packing considered by Thorup. It also shows
  that the lower bound provided by the LP relaxation is equivalent to
  the Lagrangean relaxation lower bound considered by Barahona
  \cite{b-00} and Ravi and Sinha \cite{rs-08}.
\end{itemize}

Our results help unify and simplify the different approaches to
$k$-cut via the LP relaxation and its dual.  A key motivation for this
paper is to simplify and improve the understanding of the tree packing
approach. For this reason we take a liesurely path and reprove some of
Karger's results for the sake of completeness, and to point out the
similarity of our argument for \kcut to the case of \mincut. Readers
familiar with \cite{Karger00} may wish to skip
Section~\ref{sec:mincut}.

\paragraph{Organization:} Section~\ref{sec:prelim} sets up some
basic notation and definitions. Section~\ref{sec:mincut} discusses
Karger's approach for \mincut via tree packings with some connections
to recent developments on approximately solving tree packings.
Section~\ref{sec:kcuts} describes the tree packing obtained from the
dual of the LP relaxation for \kcut and how it can be used to
extend Karger's approach to \kcut. Section~\ref{sec:kcut-lp-gap} gives
a new proof that the LP integrality gap for \kcut is $2(1-1/n)$.
In Section~\ref{sec:lp-val-for-all-k} we show that the optimum LP
solution for all values of $k$  can be characterized by a recursive
decomposition of the input graph.

\section{Preliminaries}
\label{sec:prelim}
We use $n$ and $m$ to denote the number of nodes and edges in a given
graph. For a graph $G=(V,E)$, let $\sptree(G)$ denote the set of
spanning trees of $G$. For a graph $G=(V,E)$ with edge capacities
$c: E \rightarrow \mathbb{R}_+$ the \emph{fractional spanning tree packing
number}, denoted by $\treepack(G)$, is the optimum value of a simple
linear program shown in Fig~\ref{fig:lp-treepack} whose variables are
$y_T, T \in \sptree(G)$. The LP has an exponential number of variables
but is still polynomial time solvable. There are several ways to see
this and efficient strongly combinatorial algorithms are also known
\cite{GabowM98}. We also observe that there is an optimum solution to
the LP whose support has at most $m$ trees since the number of
non-trivial constraints in the LP is at most $m$ (one per each edge).

\begin{figure}[htb]
  \begin{center}
    \begin{tcolorbox}[width=3in]
      \begin{align*}
        \max \sum_{T \in \sptree(G)} y_T & \\
        \sum_{T \ni e} y_T & \le c(e) \quad e \in E\\
        y_T & \ge 0 \quad \quad T \in \sptree(G)
      \end{align*}
    \end{tcolorbox}

  \end{center}
  \caption{LP relaxation defining $\treepack(G)$. }
  \label{fig:lp-treepack}
\end{figure}

There is a min-max formula for $\treepack(G)$ which is a special case
of the min-max formula for matroid base packing due to Tutte and
Nash-Williams. To state this theorem we introduce some notation.  For
a partition $\cP$ of the vertex set $V$ let $E(\cP)$ denote the set of
edges that cross the partition (that is, have end points in two
different parts) and let $|\cP|$ denote the number of parts of $\cP$.
A \emph{$k$-cut} is $E(\cP)$ for some partition $\cP$ such that $|\cP|\geq k$.
A \emph{cut} is a $2$-cut.
It is not hard to see that for any partition $\cP$ of the vertex set
$V$, $\treepack(G) \le \frac{c(E(\cP))}{|\cP|-1}$ since every spanning
tree of $G$ contains at least $|\cP|-1$ edges from $E(\cP)$. The
minimum over all partitions of the quantity,
$\frac{c(E(\cP))}{|\cP|-1}$, is also referred to as the strength of
$G$ (denoted by $\strength(G)$), and turns out to be equal to
$\treepack(G)$.

\begin{theorem}[Tutte and Nash Williams]
  \label{thm:tutte-nw}
  For any undirected edge capacitated graph $G$,
  $$\treepack(G) =   \min_{\cP} \frac{c(E(\cP))}{|\cP|-1}.$$
\end{theorem}

A useful and well-known corollary of the preceding theorem is given
below.

\begin{corollary}
  \label{cor:tutte-nw}
  For any graph $G$,
  $\treepack(G) \ge \frac{n}{2(n-1)} \cdot \connectivity(G)$ where
  $\connectivity(G)$ is the value of the global minimum cut of $G$ and
  $n$ is the number of nodes of $G$. If $G$ is an unweighted graph
  then $\treepack(G) \ge \frac{\connectivity(G)+1}{2}$.
\end{corollary}
\begin{proof}
  Consider the partition $\cP^*$ that achieves the minimum in the
  minmax formula. We have $c(E(\cP)) \ge |\cP^*| \connectivity(G)/2$
  since the capacity of edges leaving each part of $\cP^*$ is at last
  $\connectivity(G)$ and an edge in $E(\cP^*)$ crosses exactly two
  parts. Thus,
  $$\treepack(G) = \frac{c(E(\cP^*))}{|\cP^*|-1} \ge
  \frac{|\cP^*|\connectivity(G)}{2(|\cP^*|-1)}   \ge \frac{n}{2(n-1)} \connectivity(G)$$
  since $|\cP^*| \le n$.
  If $G$ is unweighted graph then $|\cP| \le \connectivity(G)+1$ and
  hence $\treepack(G) \ge \frac{\connectivity(G)+1}{2}$
  as desired.
\end{proof}

We say that a tree packing $y: \sptree(G) \rightarrow \mathbb{R}_+$ is
$(1-\eps)$-approximate if
$\sum_{T \in \sptree(G)} y_T \ge (1-\eps)\treepack(G)$. Note that we
typically want a compact tree packing that can either be explicitly
specified via a small number of trees of even implicitly via a data
structure representing a collection of trees. Approximate spanning
tree packings have been obtained via greedy spanning tree packings
which can be viewed as applying the multiplicative weight update
method. Recently \cite{ChekuriQ17} obtained the following result.

\begin{theorem}[\cite{ChekuriQ17}]
  \label{thm:det-tree-packing}
  There is a deterministic algorithm that, given an edge-capacitated
  undirected graph on $m$ edges and an $\eps \in (0,1/2)$, runs in
  $O(m\log^3 n/\eps^2)$ time and outputs an implicit representation of
  a $(1-\eps)$-approximate tree packing.
\end{theorem}

\section{Tree packing and \mincut}
\label{sec:mincut}
We review some of Karger's observations and results connecting tree
packings and minimum cuts \cite{Karger00} which follow relatively
easily via Corollary~\ref{cor:tutte-nw}. We rephrase his results and
arguments with a slightly different notation. Given a spanning tree
$T$ and a cut $E' \subseteq E$, 
 following Karger, we say that $T$ $h$-respects
$E'$ for some integer $h \ge 1$ if $|E(T) \cap E'| \le h$.

Karger proved that a constant fraction of trees (in the weighted
sense) of an optimum packing $2$-respect any fixed mincut. In fact
this holds for a $(1-\eps)$-approximate tree packing for sufficiently
small $\eps$. The proof, as follows, is an easy consequence of
Corollary~\ref{cor:tutte-nw} and an averaging argument. It is
convenient to view a tree packing as a probability distribution. Let
$p_T = y_T/\treepack(G)$.  We then have $\sum_T p_T = 1$ for an exact
tree packing and for a $(1-\eps)$-packing we have
$\sum_T p_T \in (1-\eps,1]$.  Let $\delta(S)$ be a fixed minimum cut
whose capacity is $\connectivity(G)$.  Let
$\ell_T = |E(T) \cap \delta(S)|$ be the number of edges of $T$ that
cross $S$. Let $q = \sum_{T: \ell_T \le 2} p_T$ be the fraction of
trees that $2$-respect $\delta(S)$. Since each tree crosses $S$ at
least once we have,
$$\sum_{T} p_T \ell_T = \sum_{T: \ell_T \le 2} p_T \ell_T + \sum_{T:
  \ell_T \ge 3} p_T \ell_T \ge q + 3 (1-\eps-q).$$ Because $y$ is a
valid packing,
$$\treepack(G) \sum_T p_T \ell_T = \sum_{T} y_T \ell_T \le
c(\delta(S)) = \connectivity(G).$$
Putting the two inequalities together and using Corollary
\ref{cor:tutte-nw},
$$3(1-\eps) - 2q \le \connectivity(G)/\treepack(G) \le 2(n-1)/n$$
which implies that
$$q \ge \frac{3}{2}(1-\eps) - (1-1/n) = \frac{1}{2} + \frac{1}{n} - \frac{3\eps}{2}.$$

If $\eps = 0$ this implies that at least half the fraction of trees
$2$-respect any minimum cut. Let $q'$ be the fraction of trees that
$1$-respect a minimum cut. One can do similar calculations as above
to conclude that
$$q' \ge 2(1-\eps) - 2(1-1/n) \ge 2(\frac{1}{n}  -\eps).$$
Thus, $q' > 0$ as long as $\eps < 1/n$. In an optimum packing there is
always a tree in the support that $1$-respects a mincut.  The
preceding argument can be generalized in a direct fashion to yield the
following useful lemma on $\alpha$-approximate cuts.

\begin{lemma}
  \label{lem:tree-crossing}
  Let $y$ be a $(1-\eps)$-approximate tree packing. Let $\delta(S)$ be
  a cut such that $c(\delta(S)) \le \alpha \connectivity(G)$ for some
  fixed $\alpha \ge 1$.  For integer $h \ge 1$ let $q_h$ denote the
  fraction of the trees in the packing that $h$-respect $\delta(S)$. Then,
  $$q_h \ge (1-\eps)(1+ \frac{1}{h}) - \frac{2 \alpha}{h} (1-\frac{1}{n}).$$
\end{lemma}

\paragraph{Number of approximate minimum cuts:} Karger showed that the
number of $\alpha$-approximate minimum cuts is at most $O(n^{2\alpha})$
via his random contraction algorithm \cite{Karger-thesis95}. He
improved the bound to $O(n^{\floor{2\alpha}})$ (for any fixed
$\alpha$) via tree packings in
\cite{Karger00}. We review the latter argument.

Given any spanning tree $T$ we can root it canonically at a fixed
vertex, say $r$.  For any cut $\delta(S)$ we can associate with $S$
the set of edges of $T$ that cross $S$, that is,
$E(T) \cap \delta(S)$.  In the other direction, a
set of edges $A \subseteq E(T)$ induces several components in $T - A$,
which induces a unique cut in $G$ where any two components of $T - A$
adjacent in $T$ lie on opposite sides of the cut. This gives a
bijection between cuts induced by edge removals in $T$,
and cuts in the graph.

Let $h = \floor{2\alpha}$. Fix an optimum tree packing $y$ and let
$\delta(S)$ be an $\alpha$-approximate mincut, that is,
$c(\delta(S)) \le \alpha \connectivity(G)$.  From
Lemma~\ref{lem:tree-crossing} and with some simplification we see that

$$q_h \ge \frac{1}{\floor{2\alpha}}\left(1 - (2\alpha - \floor{2\alpha})(1-1/n)\right).$$

Note that $q_h > 0$ and hence an easy counting argument for
approximate mincuts is the following. There is an optimum packing
whose support has at most $m$ trees. For each $\alpha$-approximate
mincut there is at least one of the trees in the packing which crosses
it at most $h$ times. Hence each $\alpha$-approximate cut can be
mapped to a tree and a choice of at most $h$ edges from that tree. The
total number of these choices is $m {n-1 \choose h} = O(m n^h)$. We
can avoid the factor of $m$ by noting that $q_h$ is a constant for
every fixed $\alpha$. We give an informal argument here. Each
approximate cut $\delta(S)$ crosses at least $mq_h$ trees at most $h$
times. If there are more than $\frac{1}{q_h}{n-1 \choose h}$ distinct
cuts and each cut induces a subset of at most $h$ edges in at least
$mq_h$ trees then we have a contradiction. Thus, the number of
$\alpha$-approximate mincuts is $O(n^{\floor{2\alpha}})$ where the
constant hidden in the big-O is $1/q_h$.

\paragraph{Minimum cut algorithm via tree packings:} Karger used tree
packings to obtain a randomized near linear time algorithm for the
global minimum cut. The algorithm is based on combining the
following two steps.
\begin{itemize}
\item Given a graph $G$ there is a randomized algorithm that outputs
  $O(\log n)$ trees in $\tilde{O}(m)$ time such that with high
  probability there is a global minimum cut that $2$-respects one of
  the trees in the packing.
\item There is a deterministic algorithm that given a graph $G$ and a
  spanning tree $T$, in $\tilde{O}(m)$ time finds the cut of minimum
  capacity in $G$ that $2$-respects $T$. This is based on a clever
  dynamic programming algorithm that utilizes the dynamic tree
  data structure.
\end{itemize}

Only the first step of the algorithm is randomized. Karger solves the
first step as follows. Given a capacitated graph $G$ and an
$\eps > 0$, he sparsifies the graph $G$ to obtain an unweighted
skeleton graph $H$ via random sampling such that (i) $H$ has
$O(n \log n/\eps^2)$ edges (ii)
$\connectivity(H) = \Theta(\log n/\eps^2)$ and (iii) a minimum cut of
$G$ corresponds to a $(1+\eps)$-approximate minimum cut of $H$ in that
the cuts induce the same vertex partition. Karger then uses greedy
tree packing in $H$ to obtain a $(1-\eps')$-tree packing in $H$ with
$O(\log n/\eps'^2)$ trees, and via Corollary~\ref{cor:tutte-nw} argues
that one of the trees in the packing $2$-respects a mincut of $G$;
here $\eps$ and $\eps'$ are chosen to be sufficiently small but fixed
constants.

We observe that Theorem~\ref{thm:det-tree-packing} can be used in
place of the sparsification step of Karger. The deterministic
algorithm implied by the theorem can be used to find an implicit
$(1-\eps)$-approximate tree packing in near linear time for any fixed
$\eps > 0$. For sufficiently small but fixed $\eps$, a constant
fraction of the trees in the tree packing $2$-respect any fixed
minimum cut. Thus, if we sample a tree $T$ from the tree packing, and
then apply Karger's deterministic algorithm for finding the smallest
cut that $2$-respects $T$ then we find the minimum cut with constant
probability. We can repeat the sampling $\Theta(\log n)$ times to
obtain a high probability bound.

Karger raised the following question in his paper. Can the dynamic
programming algorithm for finding the minimum cut that $2$-respects a
tree be made {\em dynamic}? That is, suppose $T$ is altered via edge
swaps to yield a tree $T' = T - e + e'$ where $e \in E(T)$ is removed
and replaced by a new edge $e'$. Can the solution for $T$ be updated
quickly to obtain a solution for $T$'? Note that $G$ is static, only
the tree is changing. The tree packing from
Theorem~\ref{thm:det-tree-packing} finds an implicit packing via
$\tilde{O}(m)$ edge swap operations from a starting tree
$T_0$. Suppose there is a dynamic version of Karger's dynamic program
that handles updates to the tree in amortized $g(n)$ time per
update. This would yield a \emph{deterministic} algorithm for the
global mincut with a total time of $\tilde{O}(m g(n))$. We note that
the best deterministic algorithm for capacitated graphs is
$O(mn + n^2 \log n)$ \cite{StoerW}.

\section{Tree packings for \kcut via the LP relaxation}
\label{sec:kcuts}
In this section we consider the \kcut problem. Thorup~\cite{Thorup07}
constructed a probability distribution over spanning trees which were
obtained via a recursive greedy tree packing and showed that there is
a tree $T$ in the support of the distribution such that a minimum
weight $k$-cut contains at most $2(k-1)$ edges of $T$.  He then showed
that greedy tree packing with $O(mk^3\log n)$ trees closely
approximates the ideal distribution. Via this approach he derived the
currently fastest known deterministic algorithm to find the minimum
$\kcut$ in $\tilde{O}(mn^{2k-2})$ time. This is only slightly slower
than the randomized Monter Carlo algorithm of Karger and Stein
\cite{KargerS96} whose algorithm runs in $\tilde{O}(n^{2k-2})$ time.
Thorup's algorithm is fairly simple.  However, the proofs are somewhat
complex since they rely on the recursive tree packing and its subtle
properties. Arguing that greedy tree packing approximates the
recursive tree packing is also technical.

Here we consider a different tree packing for \kcut that arises from
the LP relaxation for \kcut considered by Naor and Rabani
\cite{NaorR01}. This LP relaxation is shown in Fig~\ref{fig:lp-kcut}.
The variables are $x_e \in [0,1], e \in E$ which indicate whether an
edge $e$ is cut or not. There is a constraint for each spanning tree
$T \in \sptree(G)$; at least $k-1$ edges from $T$ need to be chosen in
a valid $k$-cut. We note that for $k > 2$ the upper bound constraint
$x_e \le 1$ is necessary.

\begin{figure}[tb]
  \begin{center}
    \begin{tcolorbox}[width=3in]
      \begin{align*}
        \min \sum_{e \in E)} c_e x_e & \\
        \sum_{e \in T} x_e & \ge k-1 \quad T \in \sptree(G)\\
        x_e & \le 1 \quad \quad e \in E \\
        x_e & \ge 0 \quad \quad e \in E
      \end{align*}
    \end{tcolorbox}

  \end{center}
  \caption{An LP relaxation for the \kcut problem from \cite{NaorR01}. }
  \label{fig:lp-kcut}
\end{figure}

The dual of the LP is given in
Fig~\ref{fig:lp-dual-kcut}.  Naor and Rabani claimed an integrality
gap of $2$ for the \kcut LP.  Their proof was incomplete and correct
proof was given in \cite{ChekuriGN06} in the context of a more general
problem called the Steiner \kcut problem. Let $\connectivity_k(G)$
denote minimum $k$-cut capacity in $G$.

\begin{theorem}[\cite{ChekuriGN06}]
  \label{thm:kcut-lp-gap}
  The worst case integrality gap of the LP for \kcut in
  Fig~\ref{fig:lp-kcut}
  is $2(1-1/n)$.
\end{theorem}

\begin{corollary}
  \label{cor:kcut-lp-gap}
  Let $(y,z)$ be an optimum solution for the dual LP for \kcut shown
  in Fig~\ref{fig:lp-dual-kcut}. Then
  $$(k-1) \sum_T y_T \ge \frac{n \kconn(G)}{2(n-1)} + z(E).$$
\end{corollary}

Note that Corollary~\ref{cor:tutte-nw} is a special case of the
preceding corollary.

\begin{remark}
  We note that the LP relaxation in Fig~\ref{fig:lp-kcut} assumes that
  $G$ is connected. This is easy to ensure by adding dummy edges of
  zero cost to make $G$ connected. However, it is useful to consider
  the general case when the number of connected components in $G$ is
  $h$ where we assume for simplicity that $h < k$ (if $h \ge k$ the
  problem is trivial). In this case we
  need to consider the maximal forests in $G$, each of which has
  exactly $n-h$ edges; to avoid notational overload we use
  $\sptree(G)$ to denote the set of maximal forests of $G$. The LP
  constraint now changes to
  $$\sum_{e \in T} x_e \ge k-h \quad T \in \sptree(G).$$
\end{remark}

\paragraph{Tree packing interpretation of the dual LP:} The dual LP
has two types of variables. For each edge $e$ there is a variable
$z_e$ and for each spanning tree $T$ there is a variable $y_T$.  The
dual seeks to add capacity $z:E \rightarrow \mathbb{R}_+$ to the
original capacities $c$, and then find a maximum tree packing
$y: \sptree(G) \rightarrow \mathbb{R}_+$ within the augmented
capacities $c+z$. The objective is
$(k-1)\sum_T y_T - \sum_{e \in E} z_e$. Note that for $k=2$, there is
an optimum solution with $z = 0$; this can be seen by the fact that
for $k=2$ the primal LP can omit the constraints
$x_e \le 1 \quad e \in E$. For $k > 2$ it may be advantageous to add
capacity to some bottleneck edges (say from a minimum cut) to increase
the tree packing value, which is multiplied by $(k-1)$.

\begin{figure}[tb]
  \begin{center}
    \begin{tcolorbox}[width=4in]
      \begin{align*}
        \text{max }             %
        & (k-1)\sum_{T \in \sptree(G)} y_T - \sum_{e \in E} z_e \\
        \text{s.t.\ }            %
        & \sum_{T \ni e} y_T \le c_e + z_e \text{ for all
          } e \in E \\
        & y_T \ge 0 \text{ for all } T \in \sptree(G)
      \end{align*}
    \end{tcolorbox}
  \end{center}
  \caption{Dual of the LP relaxation from Fig~\ref{fig:lp-kcut}. }
  \label{fig:lp-dual-kcut}
\end{figure}

Our goal is to show that one can transparently carry over the
arguments for global minimum cut via tree packings to the \kcut
setting via (optimum) solutions $y,z$ to the dual
LP. Theorem~\ref{thm:kcut-lp-gap} plays the role of
Corollary~\ref{cor:tutte-nw}. The key lemma below is
analogous to Lemma~\ref{lem:tree-crossing}.

\begin{lemma}
  \label{lem:tree-crossing-kcut}
  Let $y,z$ be an optimum solution to the dual LP for \kcut shown in
  Fig~\ref{fig:lp-dual-kcut}. Let $E' \subseteq E$ be any subset of edges
  such that $c(E') \le \alpha \kconn(G)$ for
  some $\alpha \ge 1$.  For integer $h \ge (k-1)$ let $q_h$ denote the
  fraction of the trees in the packing induced by $y,z$ that
  $h$-respect $E'$. Then,
  $$q_h \ge 1-\frac{2 \alpha (k-1)(1-\frac{1}{n})}{h+1}.$$
\end{lemma}

\begin{proof}
  Let $\treepack_k(G)$ denote $\sum_T y_T$ and let $p_T =
  y_T/\treepack_k(G)$. Let $\ell_T$ denote $|E' \cap E(T)|$.
  Thus,
  $$\sum_{T} p_T \ell_T = \sum_{T: \ell_T \le h} p_T \ell_T + \sum_{T:
    \ell_T \ge (h+1)} p_T \ell_T \ge (h+1) (1-q_h).$$
  Because $y$ is a valid tree packing in capacities $c+z$,
  $$\treepack_k(G) \sum_T p_T \ell_T = \sum_{T} y_T \ell_T \le
  c(E') + z(E') \le \alpha \kconn(G) + z(E') \le \alpha (\kconn(G) + z(E)).$$
  In the penultimate inequality of the preceding line we are using the
  fact that $\alpha \ge 1$ and that $z \ge 0$.
  Putting the the preceding inequalities together,

  \begin{equation}
    \label{eq:kcut-ineq}
    (h+1)(1-q_h) \le \frac{1}{\treepack_k(G)}\alpha(\kconn(G)+z(E)).
  \end{equation}

  Using Corollary~\ref{cor:kcut-lp-gap} and simplifying the preceding
  inequality yields,
  $$(h+1)(1-q_h) \le 2 \alpha (k-1)(1-1/n)$$
  which implies that
  $$q_h \ge 1-\frac{2 \alpha (k-1)(1-\frac{1}{n})}{h+1}.$$
\end{proof}

\begin{corollary}
  \label{cor:kcut-opt}
  Let $(y,z)$ be an solution to the dual LP. For every optimum $k$-cut
  $A \subseteq E$ there is a tree $T$ in the support of $y$ such that
  $|E(T) \cap A| \le 2k-3$.
\end{corollary}
\begin{proof}
  We apply Lemma~\ref{lem:tree-crossing-kcut} with $h=2k-3$ and
  $\alpha = 1$ and observe that $q_h > 0$ which implies the desired
  statement.
\end{proof}

\begin{corollary}
  \label{cor:kcut-approx}
  Let $(y,z)$ be a $(1-\eps)$-approximate solution to the dual LP
  where $\eps < \frac{1}{2k-1}$. For every optimum $k$-cut $A \subseteq E$
  there is a tree $T$ in the support of $y$ such that
  $|E(T) \cap A| \le 2k-2$.
\end{corollary}
\begin{proof}
  If $(y,z)$ is a $(1-\eps)$-approximate solution to the dual LP, we would have
  \begin{equation}
    \label{eq:kcut-weaker-cor}
    (k-1)\sum_T y_T \ge (1-\eps)\frac{n \kconn(G)}{2(n-1)} + z(E) \ge
    (1-\eps)\frac{\kconn(G)}{2} + z(E)
  \end{equation}
  in place of Corollary~\ref{cor:kcut-lp-gap}.

  Examining the proof of Lemma~\ref{lem:tree-crossing-kcut}, we see
  that optimality of $(y,z)$ is not used in the derivation of
  (\ref{eq:kcut-ineq}). Thus the same inequality holds for the packing
  given by $(y,z)$. Now, instead of using
  Corollary~\ref{cor:kcut-lp-gap}, we use (\ref{eq:kcut-weaker-cor})
  and simplify by setting $\alpha = 1$, to obtain the bound
  $$q_h \ge 1-\frac{2 (k-1)}{(h+1)(1-\eps)}$$
  We note that $q_h > 0$ for $h = 2k-2$ and $\eps < \frac{1}{2k-1}$.
\end{proof}

\subsection{Number of approximate $k$-cuts}
We now prove the following theorem.

\begin{theorem}
  Let $G=(V,E)$ be an undirected edge-weighted graph and let $k$ be a fixed
  integer. For $\alpha \ge 1$ the number of cuts $A$ such that
  $c(A)\leq \alpha \lambda_k(G)$ is $O(n^{\floor{2\alpha (k-1)}})$.
\end{theorem}

Let $h = \floor{2 \alpha (k-1)}$.  By
Lemma~\ref{lem:tree-crossing-kcut}, for fixed $\alpha$ and $k$, $q_h$
is a fixed and positive constant independent of $n$. Thus, for any
fixed cut $A \subseteq E$ satisfies the condition in the theorem, a constant
fraction of the trees in the packing of an optimum solution $(y,z)$
have the property they cross $A$ at most $h$ times. For a given tree
$T$ the number of distinct cuts induced by removing $h$ edges is
$O(n^h)$ for fixed $h$; as there are at most ${n-1 \choose h}$
subsets of the tree's edges, and each subset induces $f(h)$
partitions of the vertex set into at least 2 parts for some fixed function
$f(h) < h^h$. Thus if there are more than
$\frac{1}{q_h}f(h)n^{h}$ distinct cuts, we obtain a
contradiction.

In particular, each $k$-cut is a cut, we obtain the following corollary as
a special case of the above theorem.
\begin{corollary}
  Let $G=(V,E)$ be an undirected edge-weighted graph and let $k$ be a fixed
  integer. For $\alpha \ge 1$ the number of $\alpha$-approximate $k$-cuts
  is $O(n^{\floor{2\alpha (k-1)}})$.
\end{corollary}

\subsection{Enumerating all minimum $k$-cuts}

We briefly describe how to enumerate all $k$-cuts via
Lemma~\ref{lem:tree-crossing-kcut}. The argument is basically the same
as that of Karger and Thorup. First, we compute an optimum solution
$(y^*,z^*)$ to the dual LP. We can do this via the Ellipsoid method or
other ways. Let $\beta(n,m)$ be the running time to find
$(y^*,z^*)$. Moreover, if we find a basic feasible solution to the
dual LP we are guaranteed that the support of $y$ has at most $m$
distinct trees. Now Lemma~\ref{lem:tree-crossing-kcut} guarantees that
for every minimum $k$-cut $A \subseteq E$ there is a tree $T$ such
that $y(T) > 0$ and $T$ $(2k-3)$-respects $A$. Thus, to enumerate all
minimum $k$-cuts the following procedure suffices. For each of the
trees $T$ in the optimum packing we enumerate all $k$-cuts induced by
removing $h = 2k-3$ edges from $T$. With appropriate care and data
structures (see \cite{Karger00} and \cite{Thorup07}) this can be done
for a single tree $T$ in $\tilde{O}(n^{2k-3}+m)$ time. Thus the total
time for all $m$ trees in the support of $y$ is $\tilde{O}(mn^{2k-3})$
for $k > 2$. We thus obtain the following theorem.

\begin{theorem}
  \label{thm:kcut-run-time}
  For $k > 2$ all the minimum $k$-cuts of a graph with $n$ nodes and
  $m$ edges can be computed in time $\tilde{O}(mn^{2k-3} +
  \beta(m,n))$ time where $\beta(m,n)$ is the time to find an optimum
  solution to the LP for $k$-cut.
\end{theorem}

We observe that Thorup's algorithm \cite{Thorup07} runs in time
$\tilde{O}(mn^{2k-2})$. Thorup uses greedy tree packing in place of
solving the LP. The optimality of the LP solution was crucial in using
the bound of $2k-3$ instead of $2k-2$. Thus, even though we obtain a
slightly faster algorithm than Thorup, we need to find an optimum
solution to the LP which can be done via the Ellipsoid method.  The
Ellipsoid method is not quite practical. Below we discuss a different
approach.

In recent work Quanrud showed that a $(1-\eps)$-approximate solution
to the dual LP can be computed in near-linear time. We state his
theorem below.

\begin{theorem}[\cite{Quanrud18}]
  \label{thm:kcut-lp-fast}
  There is an algorithm that computes a $(1-\eps)$-approximate
  solution $(y,z)$ the dual LP in $O(m\log^3 n/\eps^2)$ time.
\end{theorem}

We observe that the preceding theorem guarantees $O(m \log^3
n/\eps^2)$ trees in the support of $y$ and also implicity stores them
in $O(m \log^3 n/\eps^2)$ space. If we choose $\eps < 1/(2k-1)$ then,
via Corollary~\ref{cor:kcut-approx}, for every minimum $k$-cut $A
\subseteq E$ there is a tree $T$ in the support of $y$ that
$(2k-2)$-respects $A$. This leads to an algorithm that in
$\tilde{O}(mn^{2k-2})$ time enumerates all minimum $k$-cuts and
recovers Thorup's running time. However, we note that the trees
generated by the algorithm in the preceding theorem are implicit, and
can be stored in small space. It may be possible to use this additional
structure to match or improve the run-time achieved by
Theorem~\ref{thm:kcut-run-time}.

\begin{remark}
  For unweighted graphs with $\tilde{O}(\frac{m}{n-k}
  \frac{1}{\eps^2})$ trees \cite{Quanrud18} guarantees a
  $(1-\eps)$-approximation.  This improves the running time to
  $\tilde{O}(mn^{2k-3})$ for unweighted graphs.
\end{remark}

We briefly discuss a potential approach to speed up the computation
futher. Recall that Karger describes an algorithm that given a
spanning tree $T$ of a graph $G$ finds the minimum cut that
$2$-respects $T$ in $\tilde{O}(m)$ time, speeding up the easier
$\tilde{O}(n^2)$ time algorithm. We can leverage this as follows. In
the case of $k > 2$ we are given $T$ and $G$ and wish to find the
minimum $k$-cut induced by the removal of at most $t$ edges where $t$
is either $2k-3$ or $2k-2$ depending on the tree packing we use.
Suppose $A$ is a set of $t-2$ edges of $T$. Removing them from $T$
yields a forest with $t-1$ components. We can then apply Kargers
algorithm in each of these components with an appropriate graph.  This
results in a running time of $\tilde{O}(mn^{t-2})$ per tree rather
than $\tilde{O}(n^t)$. We can try to build on Karger's ideas
improve the running time to find the best $3$-cut induced by
removing at most $4$ edges from $T$. We can then leverage this for
larger values of $k$.

\section{A new proof of the LP integrality gap for \kcut}
\label{sec:kcut-lp-gap}
The proof of Theorem~\ref{thm:kcut-lp-gap} in \cite{ChekuriGN06} is
based on the primal-dual algorithm and analysis of Agarwal, Klein and
Ravi \cite{akr-95}, and Goemans and Williamson \cite{gw-95} for the
Steiner tree problem. For this reason the proof is technical and
indirect. Further, the proof from \cite{ChekuriGN06} is described for
the Steiner $k$-cut problem which has additional complexity. Here we
give a different and intuitive proof for \kcut. Unlike the proof in
\cite{ChekuriGN06}, the proof here relies on optimality properties of
the LP solution and hence is less useful algorithmically. We note that
\cite{Quanrud18} uses Theorem~\ref{thm:kcut-lp-fast} and a fast
implementation of the algorithmic proof in \cite{ChekuriGN06} to
obtain a near-linear time $(2+\eps)$-approximation for \kcut.

Let $G=(V,E)$ be a graph with non-negative edge capacities
$c_e, e \in E$. We let $\deg(v) = \sum_{e \in \delta(v)} c_e$ denote
the capacitated degree of node $v$. We will assume without loss of
generality that $V=\{v_1,v_2,\ldots,v_n\}$ and that the nodes are
sorted in increasing order of degrees, that is,
$\deg(v_1) \le \deg(v_2) \le \ldots \le \deg(v_n)$.  We observe that
$\deg(v_1) + \deg(v_2) + \ldots + \deg(v_{k-1})$ is an upper bound on
the value of an optimum $\kcut$; removing all the edges incident to
$v_1,v_2,\ldots,v_{k-1}$ gives a feasible solution in which the
components are the $k-1$ isolated vertices
$\{v_1\},\{v_2\},\ldots,\{v_{k-1}\}$, and a component consisting of
the remaining nodes of the graph.

The key lemma is the following which proves the integrality gap
in a special case.

\begin{lemma}
  \label{lem:kcut-lp-round}
  Let $G$ be a connected graph and let $x^*$ be an \emph{optimum}
  solution to the \kcut LP such that $x^*(e) \in (0,1)$ for each
  $e \in E$ (in other words $x^*$ is fully fractional). Then
  $\sum_{i=1}^{k-1} \deg(v_i) \le 2(1-1/n) \sum_e c_e x^*_e.$
\end{lemma}
\begin{proof}
  Let $(y^*,z^*)$ be any fixed optimum solution to the dual
  LP. Complementary slackness gives the following two properties:
  \begin{itemize}
  \item $z^*(e) = 0$ for each $e \in E$, for if $z^*(e) > 0$ we would
    have $x^*(e) = 1$.
  \item for each $e \in E$, $\sum_{T \ni e} y^*_T = c_e$ since $x^*(e) > 0$.
  \end{itemize}
  From the second property above, and the fact that each spanning tree
  has exactly $(n-1)$ edges, we conclude that
  \begin{equation}
    \label{eq:tight-edges}
    (n-1)\sum_T y^*_T = \sum_{e \in E} c_e.
  \end{equation}

  Since the degrees
  are sorted,
  \begin{equation}
    \label{eq:small-degrees}
    \sum_{i=1}^{k-1} \deg(v_i) \le \frac{k-1}{n} \sum_{i=1}^n
    \deg(v_i) = 2 \frac{k-1}{n} \sum_e c_e.
  \end{equation}
  Putting the two preceding inequalities together,
  $$\sum_{i=1}^{k-1} \deg(v_i) \le 2 (1-\frac{1}{n}) (k-1) \sum_T y^*_T
  = 2 (1-\frac{1}{n}) \sum_e c_e x^*_e,$$
  where, the last equality is based on strong duality and the
  fact that $z^* = 0$.
\end{proof}

The preceding lemma can be easily generalized to the case when $G$ has
$h$ connected components following the remark in the preceding section
on the \kcut LP. This gives us the following.
\begin{corollary}
  \label{cor:kcut-lp-round}
  Let $G$ be a graph with $h$ connected components and let $x^*$ be an
  \emph{optimum} solution to the \kcut LP such that $x^*(e) \in (0,1)$
  for each $e \in E$. Then
  $\sum_{i=1}^{k-h} \deg(v_i) \le 2(1-1/n) \sum_e c_e x^*_e.$
\end{corollary}

Now we consider the general case when the optimum solution $x^*$ to
the \kcut LP is not necessarily fully fractional as needed in
Lemma~\ref{lem:kcut-lp-round}. The following claim is easy.

\begin{claim}
  Let $x^*(e) = 0$ where $e=uv$. Let $G'$ be the graph obtained from
  $G$ by contracting $u$ and $v$ into a single node. Then there is
  a feasible solution $x'$ to the \kcut LP in $G'$ of the same cost
  as that of $x^*$. Moreover a feasible $k$-cut in $G'$ is a feasible
  $k$-cut in $G$ of the same cost.
\end{claim}

Using the preceding claim we can assume without loss of generality
that $x^*(e) > 0$ for each $e \in E$.  Let
$F = \{e \in E \mid x^*(e) = 1\}$. Since the LP solution $x^*$ paid
for the full cost of the edges in $F$, we can recurse on $G' = G-F$
and the fractional solution $x'$ obtained by restricting $x^*$ to
$E \setminus F$.  If $G'$ is connected then $x'$ is an optimum
solution the $\kcut$ LP on $G'$, and is fully fractional, and we can
apply Lemma~\ref{lem:kcut-lp-round}. However, $G'$ can be
disconnected. Let $h$ be the number of connected components in
$G'$. If $h \ge k$ we are done since $E'$ is a feasible $k$-cut and
$c(E') \le \sum_e c_e x^*_e$.  The interesting case is when
$2 \le h < k$. In this case we apply Corollary~\ref{lem:kcut-lp-round}
based on the following claim which is intuitive and whose formal
proof we omit.

\begin{claim}
  Let $x'$ be the restriction of $x^*$ to $E \setminus F$. Then
  for any maximal forest $T$ in $G'$ we have $\sum_{e \in T} x'(e) \ge
  k-h$. Moreover, $x'$ is an optimum solution to the $\kcut$ LP
  in $G'$.
\end{claim}

From Corollary~\ref{lem:kcut-lp-round} we can find $E' \subset
E\setminus F$ such that $G'-E'$ induces a k-cut in $G'$ such that
$$c(E') \le 2(1-\frac{1}{n})\sum_{e \in E\setminus F} c_e x'_e
= 2(1-\frac{1}{n})\sum_{e \in E\setminus F} c_e x^*_e.$$
Therefore $F \cup E'$ induces a $k$-cut in $G$ and we have
that
$$c(F \cup E') = c(F) + c(E') \le \sum_{e \in F} c_e x^*_e +
2(1-\frac{1}{n})\sum_{e \in E\setminus F} c_e x^*_e \le
2(1-\frac{1}{n})\sum_{e \in E} c_e x^*_e.$$

This finishes the proof. Note that the proof also gives a very simple
rounding algorithm assuming we have an optimum solution $x^*$ for the
LP. Contract all edges with $x^*(e) = 0$, remove all edges $e$ with
$x^*(e) = 1$, and use Corollary~\ref{lem:kcut-lp-round} in the
residual graph to find the $(k-h)$ smallest degrees vertices.

\section{Characterizing the optimum LP solution}
\label{sec:lp-val-for-all-k}
We have seen that the dual of the LP relaxation for \kcut yields a
tree packing that can be used in place of Thorup's recursive tree
packing. In this section we show that the two are the same by
characterizing the optimum LP solution for a given graph through a
recursive partitioning procedure.  This yields a nested sequence of
partitions of the vertex set of the graph. This sequence is called the
principal sequence of partitions of a graph and is better understood
in the more general context of submodular functions \cite{n-91}.  We
refer the reader to Fujishige's article more on this topic
\cite{f-09}, and to \cite{c-85,k-10} for algorithmic aspects in the
setting of graphs. We also connect the LP relaxation with the
Lagrangean relaxation approach for \kcut considered by Barahona
\cite{b-00} and Ravi and Sinha ~\cite{rs-08}. Their approach is
also built upon the principal sequence of partitions. In order to keep the
discussion simple we mainly follow the notation and approach of
\cite{rs-08}.

Given $G=(V,E)$ and an edge set $A \subseteq E$ let $\numcomp(A)$
denote the number of connected components in $G-A$. Recall that the
strength of a capacitated graph $G$, denoted by $\strength(G)$ is
defined as $\min_{A \subseteq E} \frac{c(A)}{\numcomp(A) - 1}$. The
\kcut problem can be phrased as $\min_{A: \numcomp(A) \ge k} c(A)$.
However, the constraint that $\numcomp(A) \ge k$ is not
straightforward. It is, however, not hard to show that $\numcomp(A)$
is a supermodular set function over the ground set $E$. A Lagrangean
relaxation approach was considered in \cite{b-00,rs-08}. To set this
up we define, for any fixed edge set $A$, a function
$g_A: \mathbb{R}_+ \rightarrow \mathbb{R}$ as
$g_A(b) = c(A) - b(\numcomp(A) -1)$. We then obtain the function
function $g:\mathbb{R}_+ \rightarrow \mathbb{R}$ where
$g(b) = \min_{A \subseteq E} c(A) - b(\numcomp(A) - 1).$ The quantity
$g(b)$ is the attack value of the graph for parameter $b$ and was
considered by Cunningham \cite{c-85} in his algorithm to compute the
the strength of the graph.

Then, as noted in \cite{b-00,rs-08},
\begin{align*}
  \min_{A: \numcomp(A) \ge k} c(A) & \ge \max_{b \ge 0} \min_{A \subseteq E} c(A) + b(k - \numcomp(A)) = \max_{b \ge 0} g(b) + b(k-1).
\end{align*}
Thus $g'(b) = g(b) + b(k-1)$ provides a lower bound on the optimum
solution value. \cite{rs-08} describes structural properties of the
function $g$, several of which are explicit or implicit in
\cite{c-85}. We state them below.
\begin{itemize}
\item The functions $g$ and $g'$ are continuous, concave and piecewise
  linear and have no more than $n-1$ breakpoints. The function $g$ is
  non-increasing in $b$.
\item Under a non-degeneracy assumption on the graph, which is easy to
  ensure, the following holds. If $b$ is not a breakpoint then
  there is a unique edge set $A$ such that $g_A(b) = g(b)$. If $b$ is
  a breakpoint then there are exactly two edge sets $A, B$ such that
  $g_A(b) = g_B(b)$.
\item If $b_0$ is a breakpoint of $g'$ induced by edge sets $A$ and
  $B$ $\numcomp(A) > \numcomp(B)$ then $B \subset A$. In particular
  $A \setminus B$ is contained in some connected component of $G'=(V,E
  \setminus B)$.
\item Let $b_0$ be a breakpoint of $g'$ induced by edge set $A$. Then
  the next breakpoint is induced by the edge set which is the solution
  to the strength problem on the smallest strength component of
  $G'=(V,E \setminus A)$.
\end{itemize}

The above properties show that the breakpoints induce a sequence of
partitions of $V$ which are refinements.  Alternatively we consider
the sequence of edge sets $A_1,A_2, \ldots,$ obtained by the following
algorithm. We will assume that $G$ is connected.  Let
$A_0=\emptyset$. Given $A_i$ we obtain $A_{i+1} \supseteq A_i$ as
follows.  Let $G_i = (V,E\setminus A_{i-1})$. If $G_i$ has no edges we
stop. Otherwise let $C_{i+1}$ be the minimum strength connected
component of $G_i$ and $B_{i+1}$ be a minimum strength edge set of
$C_{i+1}$. We define $A_{i+1} = A_i \cup B_{i+1}$. The process stops
when $A_h = E$. Let $\cP_i$ denote the partition of $V$ induced by
$A_i$. Note that $\cP_{i+1}$ is obtained from $\cP_i$ by replacing
$C_{i+1}$ by a minimum strength partition of $C_{i+1}$, thus
$\cP_{i+1}$ is a refinement of $\cP_i$ and $\cP_h$ consists of
singleton nodes. Note that Thorup's ideal tree packing is also
based on the same recursive decomposition.

Ravi and Sinha obtained a $2$-approximation for \kcut as follows.
Given the preceding decomposition of $G$ they consider the smallest
$j$ such that $|\cP_j| \ge k$. If $|\cP_j| = k$ they output it and can
argue that it is an optimum solution. Otherwise they do the
following. Recall $\cP_j$ is obtained from $\cP_{j-1}$ by replacing
the component $C_j$ in $G - A_{j-1}$ by a minimum strength
decomposition of $C_j$. Let $k' = k - |\cP_{j-1}|$. Consider the
minimum strength partition of $C_j$ and let $H_1,H_2,\ldots,H_{k'}$ be
the connected components of the partition with the smallest shores.
Output the cut $A_{j-1} \cup (\cup_{\ell=1}^{k'} \delta(H_\ell))$.

\paragraph{An optimum LP solution from the decomposition:} Given $k$,
as before let $j$ be the smallest index such that
$\numcomp(A_j) \ge k$. Let $k' = k - |\cP_{j-1}|$. We consider the
following solution to the LP:
\begin{itemize}
\item $x(e) = 1$ for each $e \in A_{j-1}$.
\item $x(e) = \alpha$ for each $e \in A_j \setminus A_{j-1}$, where
  \begin{math}
    \alpha = \frac{k - \numcomp(A_{j-1})}{\numcomp(A_j) -
      \numcomp(A_{j-1})}.
  \end{math}
\item $x(e) = 0$ for each $e \in E \setminus A_j$.
\end{itemize}

\begin{lemma}
  \label{lem:feasibleLPsolution}
  The solution $x$ is feasible and has objective value
  \begin{align*}
    \bar{c}(A_{j-1}) + \alpha \bar{c}(B_{j}) %
    =                           %
    \bar{c}(A_{j-1}) +
    \left( k - \numcomp_{j-1} \right) \lambda_j,
  \end{align*}
  where we denote $\bar{c}(A) = \sum_{e \in A} c(e)$.
\end{lemma}

\begin{proof}
  Let $T$ be any spanning tree. We want to show that
  $\sum_{e \in T} x(e) \geq k-1$.  For each $j$, let
  $\numcomp_j = \numcomp(A_{j})$, and let $\ell_j = \abs{T \cap
    A_j}$. Then $T$ has $\ell_{j-1}$ edges of value $x(e) = 1$, and
  $\ell_j - \ell_{j-1}$ edges of value $\alpha$.  We have
  \begin{align*}
    \sum_{e \in T} x(e)         %
    &=                           %
      \ell_{j-1} +
      (\ell_j - \ell_{j-1}) \alpha
    \\
    &\geq                     %
      \numcomp_{j-1} - 1
      +                         %
      (\numcomp_j - \numcomp_{j-1})
      \alpha
      = k - 1,
  \end{align*}
  where we observe that the RHS of the first line is decreasing in
  both $\ell_j$ and $\ell_{j-1}$,
  \begin{math}
    \ell_j \geq \numcomp_{j} - 1,
  \end{math}
  and
  \begin{math}
    \ell_{j-1} \geq \numcomp_{j-1} - 1.
  \end{math}
  To calculate the objective value, we have
  \begin{align*}
    \sum_{e \in E} x(e)         %
    =                           %
    \sum_{e \in A_{j-1}} c(e)
    +                           %
    \sum_{e \in A_j \setminus A_{j-1}} \alpha c(e) %
    =                                             %
    \bar{c}(A_{j-1}) + \alpha\bar{c}(B_j)
  \end{align*}
\end{proof}

The harder part is:

\begin{lemma}
  The solution $x$ attains the optimum value to the LP relaxation.
\end{lemma}

\begin{proof}
  We prove the claim by constructing a dual solution equal value. See
  Figure \ref{fig:lp-dual-kcut} for the dual LP.

  Recall the definitions of $\cP_i$, $A_i$, $B_i$, and $C_i$ from above.
  For each $i$, let $\numcomp_i = \numcomp(A_i) = |\cP_i|$ be
  the number of components in the $i$th partition.  Let
  $\lambda_1 < \lambda_2 < \cdots < \lambda_j$ be the strengths of the
  components $C_1,C_2,\dots,C_j$. Let $Q_i$ be the partition on $C_i$
  corresponding to $B_i$.  An ideal tree packing, following
  \cite{Thorup08}, is a convex combination of trees
  $p: \sptree(G) \to [0,1]$ s.t.\ $\sum_T p_T = 1$ with the following
  properties.
  \begin{enumerate}
  \item For each $i$, every tree $T$ supported by $p$ induces a tree
    in the graph $G / \cP_i$ obtained by contracting each component of
    $\cP_i$.
  \item For each $i$ and each edge $e \in B_i$, $p$ induces a load of
    $1/\lambda_i$ on $e$.
  \end{enumerate}
  Every graph has an ideal tree packing, and (for example) can be
  constructed recursively as follows. For each $C_i$, we write each
  $B_i$ as a sum of $\lambda_i$ (units of fractional) trees in
  $C_i / Q_i$ (which holds because $B_i$ is a minimum strength cuts),
  and scale it down to a distribution of trees in $C_i / Q_i$ with
  load $1 / \lambda_i$ on each edge in $B_i$. An ideal tree packing
  now corresponds to the distribution where we take the union of one
  sampled spanning tree from (the distribution of) each $C_i / Q_i$.

  Let $p: \sptree(G) \to [0,1]$ be an ideal tree packing. To construct
  our dual solution, we define nonnegative edge potentials
  $z(e) \geq 0$ and a tree packing $y(t) \geq 0$ (packing into
  $c + z$) s.t.\
  \begin{align*}
    y_T
    &=                     %
      \lambda_j p(T)       %
      \hspace{17.5pt}\text{for all } T \in \sptree(G), \\
    c(e) + z(e)                 %
    &=                     %
      \begin{cases}
        \frac{\lambda_j}{\lambda_i} c(e) &\text{for all } e
        \in
        B_i \text{ for } i < j \\
        c(e) &\text{otherwise.}
      \end{cases}
  \end{align*}
  We first claim that $(y,z)$ is feasible in the dual LP; that is, $y$
  is a feasible tree packing w/r/t the augmented capacities $c + z$.
  Observe that for any edge $e$, $y$ uses capacity $\lambda_j$ times
  the capacity by $p$. We need to show the capacity used by $y$ along
  any edge $e$ is at most $c(e) + z(e)$. We have two cases.
  \begin{enumerate}
  \item If $e \in B_i$ for some $i < j$, then $p$ uses capacity
    $\frac{c(e)}{\lambda_i}$. In turn, $y$ uses capacity
    \begin{math}
      \frac{\lambda_j}{\lambda_i} c(e).
    \end{math}
    By choice of $z(e)$, we have
    \begin{math}
      c(e) + z(e) = \frac{\lambda_j}{\lambda_i} c(e),
    \end{math}
    as desired.
  \item If $e \in E \setminus A_{j-1}$, then $p$ uses capacity at most
    $\frac{c(e)}{\lambda_j}$. In turn, $y$ uses capacity at most
    $\frac{\lambda_j}{\lambda_j} c(e)$. But $\lambda_j \leq \lambda_j$,
    so the capacity used by $y$ is $\leq c(e)$.
  \end{enumerate}
  We now analyze the objective value of our dual solution. We
  first observe that since each tree supported by $y$ is a tree in
  $G / \cP_j$, we have
  \begin{align*}
    (k-1) \sum_T y_T            %
    &=                           %
      \frac{k-1}{\numcomp_j-1} \sum_T y_T {| T \cap A_j |}
      =                           %
      \frac{k-1}{\numcomp_j-1} \sum_{e \in A_j} \sum_{T \ni e} y_T
    \\
    &=                           %
      \frac{k-1}{\numcomp_j-1} \lambda_j \sum_{i \leq j}
      \frac{1}{\lambda_i} \sum_{e \in B_i}  c(e)
      =                           %
      \frac{k-1}{\numcomp_j-1} \lambda_j \sum_{i \leq j}
      {\numcomp_i - \numcomp_{i-1}}
    \\
    &=                           %
      \frac{k-1}{\numcomp_j-1} \lambda_j ({\numcomp_j - 1})
      =                     %
      (k-1) \lambda_j
      .
  \end{align*}
  On the other hand, when subtracting out the augmented capacities, we
  have
  \begin{align*}
    \sum_{e \in E} z(e)         %
    &= %
      \sum_{i < j} \sum_{e \in B_i} \left({\frac{\lambda_j}{\lambda_i} - 1}\right)
      c(e) %
      =                           %
      \lambda_j
      \left(                       %
      \sum_{i < j} \frac{1}{\lambda_i} \sum_{e \in B_i} c(e)
      \right)                       %
      -                           %
      \sum_{e \in A_{j-1}} c(e)
    \\
    &=                           %
      \lambda_j \sum_{i < j} \left( \numcomp_i - \numcomp_{i-1}  \right)
      -                           %
      \sum_{e \in A_{j-1}} c(e)       %
      =                           %
      \lambda_j \left( \numcomp_{j-1} - 1 \right) - \sum_{e \in A_{j-1}} c(e)
  \end{align*}
  Thus the total objective value of our solution, as a function of
  $\lambda_j$, is
  \begin{align*}
    (k-1) \sum_T y_T - \sum_{e \in E} z(e) %
    =                                      %
    \left( k - \numcomp_{j-1}
    \right)                     %
    \lambda_j
    +                           %
    \sum_{e \in A_{j-1}} c(e),
  \end{align*}
  as desired.
\end{proof}

\begin{remark}
  One can also verify the optimality of $(x,y,z)$ in the proof above
  by complimentary slackness conditions. Recall that $x$ and $(y,z)$
  satisfy the complimentary slackness conditions if
  \begin{enumerate}
  \item $z_e > 0$ only if $x_e = 1$.
  \item $y_T > 0$ only if $\sum_{e \in T} x_e = k-1$.
  \item $x_e > 0$ only if $\sum_{T \ni e} y_e = c_e + z_e$.
  \end{enumerate}
  We address these individually.
  \begin{enumerate}
  \item $z_e > 0$ only if $e \in B_i$ for some $i < j$. In this case,
    $e \in A_{j-1}$ so $x_e = 1$.
  \item If $y_T > 0$ then $T$ is in the support of the ideal tree
    packing. In particular, $T$ contains exactly $\numcomp_{j} - 1$
    edges from $A_{j}$ and $\numcomp_{j} - 1$ edges from $A_{j-1}$, so
    we have
    \begin{align*}
      \sum_{e \in T} x_e        %
      =                         %
      \sum_{e \in T \cap A_{j-1}} 1 +
      \sum_{e \in
      T \cap (A_j \setminus A_{j-1})} %
      \frac{k - \numcomp_{j-1}}{\numcomp_j - \numcomp_{j-1}} %
      =                         %
      \numcomp_{j-1} - 1 + k - \numcomp_{j-1} %
      =                         %
      k - 1,
    \end{align*}
    as desired.
  \item If $x_e > 0$, then $e \in B_i$ for some $i \leq j$, so $y$
    uses $\frac{\lambda_j}{\lambda_i} c(e) = c(e) + z(e)$ units of
    capacity of $e$, as desired.
  \end{enumerate}
\end{remark}

\paragraph{Implications of the characterization:} We now outline some
implications of the preceding characterization of the optimum LP
solution.

Ravi and Sinha showed that Lagrangian relaxation lower bound is no
weaker than the one provided by LP relaxation. Here we show that they
are equivalent.
\begin{theorem}
The Lagrangian relaxation value is the same as the LP
value.
\end{theorem}
\begin{proof}
Let $\numcomp_i=\numcomp(A_i)=|\cP_i|$ be the number of components after
removing $A_i$, and $\lambda_i$ to be the critical value of $|\cP_i|$.
If $\numcomp_j=k$ for some $j$, then the Lagrangian relaxation value is the
min $k$-cut value. Hence it matches the LP value.
Otherwise, assume $\numcomp_{j-1}<k<\numcomp_j$.
In this case, one can see the function $g'$ maximizes at $\lambda_j$.
Since $g'$ is piecewise linear, we have
\[
g'(\lambda_j) = \bar{c}(A_{j-1})-\lambda_j(\numcomp_{j-1}-1)+\lambda_j(k-1) = \bar{c}(A_{j-1})+(k-\numcomp_{j-1})\lambda_j
.\]
This is precisely the value of the LP in \autoref{lem:feasibleLPsolution}.
\end{proof}
The preceding also gives yet another proof that the integrality gap
of the LP is $2(1-1/n)$.

\medskip

Second, as we saw, for any value of $k$, an optimum dual solution to
the \kcut LP can be derived from the ideal tree packing
\cite{Thorup07,Thorup08}. The last issue is the connection between
greedy tree packing and the dual LP. At the high-level it is tempting
to conjecture that greedy tree packing is essentially approximating
the dual LP via the standard MWU approach. Proving the conjecture
formally may require a fair amount of technical work and we leave it
for future work. We believe that some insights obtained in
\cite{Quanrud18} could be useful in this context; \cite{Quanrud18}
recasts the LP relaxation for \kcut into a pure covering LP, and the
dual as a pure packing LP that involves packing forests.

\bibliographystyle{plain}
\bibliography{kcut}

\end{document}